\documentclass[a4paper, 11pt]{article}

\usepackage{fullpage}
\usepackage[english]{babel}
\usepackage[latin1]{inputenc}
\usepackage{mathrsfs}
\usepackage{amsfonts}
\usepackage{amssymb}
\usepackage{amsmath}
\usepackage{amsthm}
\usepackage{graphicx}
\usepackage[usenames,dvipsnames]{color}
\usepackage[colorlinks=true,citecolor=blue,linkcolor=magenta]{hyperref}
\usepackage{lmodern}
\usepackage{microtype}
\usepackage{braket}
\usepackage{tabularx,colortbl}
\usepackage{pifont}
\usepackage{mathtools}
\usepackage{color}
\usepackage{dsfont}
\usepackage{algpascal}
\usepackage{algorithm} 
\usepackage{algpseudocode}
\usepackage{chngcntr}
\usepackage{complexity}
\usepackage{float} % loaded by algorithm
\usepackage[table, x11names]{xcolor}
\usepackage{epstopdf}
\usepackage{braket}
\usepackage{caption}   
\usepackage{subfig}   
\usepackage[normalem]{ulem}
\usepackage[titletoc,title]{appendix}

\newtheorem{theorem}{Theorem}
\newtheorem{corollary}[theorem]{Corollary}

\newtheorem{lemma}{Lemma}

\newcommand{\comment}[1]{} 
\newcommand{\Norm}[1]{\left\lVert#1\right\rVert} 
\newcommand{\norm}[1]{\lVert#1\rVert} 
\newcommand{\idty}[1]{\mathbb{1}} 
\newcommand{\ovsqrt}[1]{\frac{1}{\sqrt{2}}}
 
\newcommand{\tr}[1]{\mathrm{Tr}}

\newcommand{\QEX}{\mathrm{QEX}}
\newcommand{\EX}{\mathrm{EX}}
\newcommand{\reals}{\mathbb{R}}
\newcommand{\Ex}{\operatorname*{\mathbb{E}}\nolimits}
\def\Pr{\mathop{\rm Pr}\nolimits}

\usepackage[backend=biber, natbib=true, bibstyle=alphabetic,citestyle=alphabetic, doi=true, sorting=nyt, giveninits]{biblatex}
\title{Learning DNFs under product distributions via $\mu$--biased quantum Fourier sampling}

\begin{document}

\author{Varun Kanade\thanks{University of Oxford. Email: \href{mailto:varunk@cs.ox.ac.uk}{varunk@cs.ox.ac.uk}} \and Andrea Rocchetto\thanks{University of Oxford and University College London. Email: \href{mailto:andrea.rocchetto@spc.ox.ac.uk}{andrea.rocchetto@spc.ox.ac.uk}} \and Simone Severini\thanks{University College London and Shanghai Jiao Tong University. Email: \href{mailto:s.severini@ucl.ac.uk}{s.severini@ucl.ac.uk}}}

\date{}
\maketitle

\begin{abstract}
	We show that DNF formulae can be quantum PAC-learned in polynomial time
	under product distributions using a quantum example oracle. The current best
	classical algorithm	runs in superpolynomial time. Our result extends the work by Bshouty and
	Jackson (1998) that proved that DNF formulae are efficiently learnable under
	the uniform distribution using a quantum example oracle. Our proof is based
	on a new quantum algorithm that efficiently samples the coefficients of a
	$\mu$--biased Fourier transform.
\end{abstract}

\section{Introduction}

Whether the class of Boolean functions that can be expressed as polynomial size
formulae in disjunctive normal form (DNF) is \emph{probably approximately
correct} (PAC) learnable in polynomial time, or not, is one of the central
unresolved questions in the PAC learning framework introduced
by Valiant~\cite{valiant1984theory}. Currently, the best classical algorithm for this
problem has running time $2^{\tilde{O}(n^{1/3})}$~\cite{klivans2001learning}. A number of variants of this problem have been studied by relaxing the
requirements---primarily by requiring the learning algorithm to work only when
the underlying distribution is uniform, as well as providing greater power to
the learning algorithm, e.g. access to a membership
query oracle~\cite{jackson1994efficient, awasthi2013learning}, a labeled random
walk~\cite{bshouty2005learning}, or a quantum example
oracle~\cite{bshouty1998learning}.

Two settings in which it is possible to show learnability of DNF formulae under specific
assumptions are particularly relevant to our work. First, when the distribution
is uniform, a quasi-polynomial (in fact an $n^{O({\mathrm{log}(n)})}$)
algorithm is known~\citep{verbeurgt1990learning}. Second, in the
\textit{membership query} (MQ) model, where the learner can query an
oracle for the value of the unknown function at a given point in the domain,
Jackson gave a polynomial time algorithm for learning DNFs that works with respect to
both the uniform and product distributions~\cite{jackson1997efficient}.

Bshouty and Jackson extended the PAC learning framework to a setting where learners have
access to quantum resources~\cite{bshouty1998learning}. Key results in the area
of quantum learning theory are reviewed in a recent survey by
Arunachalam and de Wolf~\cite{arunachalam2017survey}. The two main factors that distinguish a quantum
PAC-learner from a classical one are the ability to query an oracle that can
provide examples in quantum superposition and access to a quantum computer to
run the learning algorithm. 
Two measures of interest in the PAC learning framework are \emph{sample complexity}--the
worst-case number of examples required to learn a class of functions, and
\emph{time complexity}--the worst-case running time of a learner for that
concept class; clearly the sample complexity is at most the time complexity,
but may in principle be significantly smaller. It has been shown that the
quantum PAC model gives only a constant factor advantage in terms of sample
complexity with respect to the classical
analogue~\citep{arunachalam2016optimal}. Certain results suggest that there is a separation between the 
classical and quantum PAC model when considering the time complexity of
some learning problems. When learning with respect to the uniform distribution,
the class of polynomial-size DNF formulae~\citep{bshouty1998learning} and
$k$-juntas~\citep{atici2007quantum} under the uniform distribution are known to
be efficiently quantum PAC-learnable (note that the learnability of $k$-juntas
is implied by the result on DNFs). In the classical setting, in both these
cases, current best known algorithms are quasi-polynomial time algorithms
(assuming $k = \omega(1)$). While no formal hardness results are
		known in the standard PAC framework, it would be highly surprising if a
polynomial time algorithm for these algorithms in the classical setting was
discovered. Recently, Daniely and Shalev-Shwartz proved PAC learning hardness
for DNFs under the assumption that no fixed polynomial-time algorithm can
refute random k-SAT formulae~\cite{daniely2016complexity}.
Information-theoretic lower bounds are known in more restricted models, such as
the statistical query model, which suggest these classes cannot be learned in
polynomial time~\citep{blum1994weakly}. In the context of learning in the
presence of noise, \citet{cross2015quantum} proved that parity functions under
the uniform distribution can be efficiently learned using a quantum example
oracle.  Classically, the problem is widely believed to require subexponential,
but superpolynomial, time~\citep{blum2003noise,lyubashevsky2005parity}. The
result of Cross, Smith, and Smolin was extended to the \emph{learning with errors} (LWE) problem and to
more general error models in~\cite{grilo2019learning}.

\subsection{Overview of our results}
\noindent
We show that DNF formulae under constant-bounded product distributions can be learned in
polynomial time in the quantum PAC model. Our proof builds on the work by Feldman
 for learning DNFs under the product distribution
using membership queries~\cite{feldman2012learning}. Feldman's proof is in turn based on a result
by Kalai, Samordintsky, and Teng that shows that DNFs can be approximated by heavy
low-degree Fourier coefficients alone~\cite{kalai2009learning}. Notably, Feldman's result also applies
to learning settings where the examples are drawn from a product distribution,
i.e. a distribution that factorises over the elements of the input vector. 

The only part of Feldman's algorithm that makes use of membership queries is
the subroutine that approximates the Fourier spectrum of $f$. The approximation
is obtained using the \textit{Kushilevitz-Mansour} (KM)
algorithm~\citep{kushilevitz1993learning}, for the case of uniform
distributions, and the \textit{extended Kushilevitz-Mansour} (EKM)
algorithm~\citep{kalai2009learning}, for the case of product distributions.
Bshouty and Jackson showed that it is possible to approximate the Fourier
coefficients of $f$ using quantum Fourier sampling,  a technique introduced
by~\cite{bernstein1997quantum}. For a Boolean function ($\pm 1$-valued)
	the sum of the squares of the Fourier coefficients is 1 and as a result the
	Fourier transform represents a distribution over the Fourier basis weighted
	by the square of the corresponding Fourier coefficient; this distribution is
	denoted by $\hat{f}^2$.  The technique of~\cite{bernstein1997quantum} allows one to sample efficiently from this distribution
using the \textit{Quantum Fourier Transform} (QFT).

In order to extend the result by~\citet{bshouty1998learning} for learning
	under product distributions, it is sufficient to find a quantum technique to
	sample from the Fourier distribution corresponding to the underlying product
	distribution.  As in the case of the uniform distribution, the sum of the
	squared Fourier coefficients is $1$ and the basis functions are orthonormal
	with respect to the inner product defined with respect to the product
	distribution~\citep{bahadur1961representation,furst1991improved}.  We will
call the resulting Fourier transform the $\mu$-\textit{biased
Fourier transform}, where $\mu$ is the product distribution.

In this work we introduce the $\mu$-biased \textit{quantum} Fourier transform.
We show the validity of our construction in two steps. First, we explicitly
construct a unitary operator that implements the single qubit transform. Then we argue
that this construction can be efficiently implemented on a quantum circuit with
logarithmic overhead. By exploiting the factorisation of product distributions,
we show how to build an $n$-qubit transform as a tensor product of $n$ single
qubit transforms.
Our algorithm does not require prior knowledge of the parameter  $\mu$ that characterises the product distribution. This can be estimated efficiently via sampling and we show that the error introduced due to sampling can be suitably bounded.

The main technical contribution of this paper is a quantum algorithm to
approximate the heavy, $\mu$--biased, low-degree Fourier spectrum of $f$ for
constant-bounded product distributions without using membership queries (recall
that membership queries are necessary in Feldman's classical algorithm). This
can be interpreted as a quantum version of the EKM algorithm for approximating
the low-degree Fourier coefficients of $f$. We provide rigorous upper bounds on
the scaling of the algorithm using the Dvoretzky--Kiefer--Wolfowitz theorem, a
concentration inequality that bounds the number of samples required to estimate
a probability distribution in the infinity norm. The learnability of DNFs under
the product distribution immediately follows from an application of the quantum
EKM algorithm to Corollary~5.1 in~\cite{feldman2012learning}. 

\subsection{Related work}
\noindent
The learnability of DNF formulae under the uniform distribution using a quantum
example oracle was first studied by~\citet{bshouty1998learning}, who in the same
paper, also introduced the quantum PAC model.  Their approach to learning DNF
was built on the \textit{harmonic sieve} algorithm previously developed
by~\citet{jackson1997efficient}. Jackson's algorithm exploits
a property of DNF formulae known as concentration of Fourier spectrum. More specifically,
Jackson used the fact that for every $s$-term DNF and for
every probability distribution $\mathcal{D}$, there exists a parity $\chi_a$
such that $|\Ex_\mathcal{D} [f \chi_a] |\geq 1/(2s +
1)$. This implies that for every $f$ and
$\mathcal{D}$ there exists a parity that weakly approximates $f$. In the
harmonic sieve algorithm, the boosting algorithm of~\cite{freund1995boosting},
is then used to turn the weak learner into a strong one.  The only part of the
harmonic sieve algorithm that requires membership queries is the KM algorithm
used to find the weakly approximating parity function. Bshouty and Jackson
consider the setting where the examples are given by a quantum example oracle
and replace the KM algorithm with quantum Fourier
sampling due to \citet{bernstein1997quantum}.
 
\citet{jackson2002quantum}~studied the learnability of DNFs in the quantum
membership model~(where the quantum example oracle is replaced by an oracle
that returns $f(x)$ for a given $x$). By using the quantum Goldreich-Levin
algorithm developed by~\citet{adcock2002quantum}, they were able
to obtain a better bound on the query complexity with respect to the best
classical algorithm. We recall that the classical KM algorithm can be derived
from the Goldreich-Levin theorem, an important result that reduces the
computational problem of inverting a one-way function to the problem of
predicting a given hard-predicate associated with that
function~\cite{goldreich1989hard}. The result in~\cite{adcock2002quantum} shows that this
reduction can be obtained more efficiently when considering quantum functions
and quantum hard--predicates. A different quantum implementation of the
Goldreich-Levin algorithm was given in~\cite{montanaro2010quantum}.

\subsection*{Organisation}
\noindent
We describe notation and important background concepts in
Section~\ref{sec:prelims}. In Section~\ref{sec:muQFT} we define the
$\mu$--biased quantum Fourier transform and discuss some of its properties. In
Section~\ref{sec:muspectrum} we introduce an efficient quantum algorithm to
sample from the Fourier coefficients of the $\mu$--biased Fourier transform and
show how this can be used to prove the PAC-learnability of DNF formulae under
product distributions. We conclude in Section~\ref{app:error} we bound the
error introduced by approximating $\mu$.

\section{Preliminaries}
\label{sec:prelims}
\noindent
\subsection{Notation}
\noindent
We denote vectors with lower-case letters. For a vector $x\in \mathbb{R}^n$, let $x_i$ denotes the $i$-th element of $x$. If $x$ is sparse we can describe it using only its non-zero coefficients. We call this the \textit{succinct representation} of $x$. For an integer $k$, let $[k]$ denote the set $\{1,\dots, k\}$. We use the following standard norms: The $\ell_0$ ``norm'' $\norm{x}_0 = |\{i\in[k] | x_i \neq 0\}|$, the $\ell_2$ norm $\norm{x}_2 = \sqrt{\sum_{i\in [k]} x_i^2}$, and the $\ell_\infty$ norm $\norm{x}_\infty = \max_{i\in [k]} \{|x_i|\}$. 

Let $f:\mathbb{R}\rightarrow \mathbb{R}^+$ and $g:\mathbb{R}\rightarrow \mathbb{R}^+$. We use $f(n)=O(g(n))$ to indicate that the asymptotic scaling of $|f|$ is upper-bounded, up to a constant factor, by $g(n)$. If the bound is not asymptotically tight we write $f(n) = o(g(n))$. Similarly, $f(n) = \Omega(g(n))$ indicates that the asymptotic scaling of $|f|$ is lower-bounded, up to a constant factor, by $g$. If the bound is not asymptotically tight we write $f(n) = \omega(g(n))$. The notation $f(n) = \Theta(g(n))$ indicates that $f$ is bounded both above and below by $g$ asymptotically. The notations $\tilde{O}(g(n))$ and $\tilde{\Omega}(g(n))$ hide logarithmic factors.

The probability that an event $E$ occurs is denoted by $\Pr[E]$. Given a set $A$ the indicator function $\mathbf{1}_{A} : A \rightarrow \{0,1\}$ takes values $\mathbf{1}_{A}(x) = 0$ if $x \notin A$ and $\mathbf{1}_{A}(x) = 1$ if $x \in A$.
Let $X$ be a continuous or discrete random variable.
By abuse of notation, we write  $\mathcal{D}_X(X = x)$ to indicate both the \textit{probability density function} (pdf) of the distribution at $x$ (for continuous variables) and the probability that $X = x$ (for discrete variables).
For a realisation $x$ we often write $\mathcal{D}_X(x)$ rather than $\mathcal{D}_X(X = x)$.
We continue the presentation assuming $X$ is a continuous variable as the notation extends straightforwardly to the discrete case.  
If a pdf $\mathcal{D}_X$ depends on a parameter $\mu$ we write $\mathcal{D}_X (x;\mu)$.
For convenience we often drop the subscript and write the pdf as $\mathcal{D} (x)$. Note that using this notation $\mathcal{D} (x)$ and $\mathcal{D}(y)$ denote different pdfs, referring to two different random variables $X$ and $Y$, respectively.
The notation $X \sim \mathcal{D}$ indicates that the random variable has pdf $\mathcal{D}$ while the notation $x\sim \mathcal{D}$ indicates that $x$ is sampled according to $\mathcal{D}$. 
For a finite sequence $S = (x^{(1)}, \dots, x^{(n)})$, $S\sim\mathcal{D}^n$ indicates that the sequence $S$ is \textit{independent and identically distributed} (i.i.d.) according to $\mathcal{D}$. 
The expected value of a random variable $f(X)$ is denoted as $\Ex_{X\sim \mathcal{D}} [f(X)] = \int f(x) d\mathcal{D}(x) = \int f(x) \mathcal{D}(x) dx$, where we assumed that $X$ has density $\mathcal{D}$ with respect to the Lebesgue measure.
When $\mathcal{D}$ is the uniform distribution we omit the distribution in the subscript and write $\Ex [\, \cdot \,]$.
We often use $\Ex_{\mathcal{D}} [\, \cdot \,]$ to indicate $\Ex_{X\sim\mathcal{D}} [\, \cdot \,]$. 
By abuse of notation, and only when there is only a single random variable involved in the discussion, we write $\Ex_{\mu} [\, \cdot \,]$ to indicate $\Ex_{X \sim \mathcal{D}(\cdot ; \mu)} [ \, \cdot \,]$. We use similar notation for $\Pr$.

\subsection{Fourier analysis over the Boolean cube}
\noindent
Let $x\in \{-1,1\}^n$ and let $f$ and $g$ be real-valued functions defined over
the Boolean hypercube $f,g:\{-1,1\}^n \rightarrow \reals$. The space of real
functions over the Boolean hypercube is a vector space with inner product $
\langle f,g \rangle = \frac{1}{2^n}\sum_{x\in \{0,1\}^n} f(x)g(x) =
\Ex[f \cdot g]$ where the expectation is taken uniformly over all $x\in
\{0,1\}^n$. A \textit{parity function} $\chi_a:\{-1,1\}^n \rightarrow \{-1,1\}$
labels a $x\in \{-1,1\}^n$ according to a characteristic vector $a \in
\{0,1\}^n$ and is defined as $\chi_a (x) = (-1)^{a\cdot x}$ where $a\cdot x = \sum_{i=1} ^n a_i x_i$. The set of parity
functions $\{\chi_a \}_{a\in\{0,1\}^n}$ forms an orthonormal basis for the
space of real-valued functions over the Boolean hypercube. This fact implies that we
can uniquely represent every function $f$ as a linear combination of parities,
the \textit{Fourier transform} of $f$.  The linear coefficients, known as the
\textit{Fourier coefficients}, are given by the projections of the function
into the parity base and are denoted with $\hat{f}(a) = \langle f,\chi_a
\rangle = \mathbb{E}[f(x)\chi_a(x)]$. The set of Fourier
coefficients is called the \textit{Fourier spectrum} of $f$ and is denoted by
$\hat{f}$, which can also be seen as a $2^n$ dimensional vector in
$\mathbb{R}^{2^n}$. For a set $S \subseteq \{0,1\}^n$, $\hat{f}(S)$ denotes the
vector of all Fourier coefficients with indices in $S$. The \textit{degree} of
a Fourier coefficient $\hat{f}(a)$ is $\norm{a}_0$. Let $B_d = \{a \in
\{0,1\}^n \, | \, \norm{a}_0 \leq d\}$. We denote by $\hat{f}(B_d)$ vector of
all degree--$\leq d$ coefficients of $f$. The squared Fourier coefficients are
related by Parseval's identity $\Ex[f^2] = \sum_a \hat{f} (a) ^2 =
\norm{\hat{f}}_2 ^2$. This implies that for any $f:\{-1,1\}^n \rightarrow
[-1,1]$, $\sum_a \hat{f} (a) ^2  \leq 1$. The equality holds if $f$ is
Boolean--valued, i.e. if $f:\{-1,1\}^n \rightarrow \{-1,1\}$ and therefore $\hat{f} ^2$ forms a probability distribution. Let $\gamma>0$, we say that a Fourier coefficient $\hat{f}(a)$ is
$\gamma$-\textit{heavy} if it has large magnitude $|\hat{f}(a)|>\gamma$. By Parseval's identity the number of $\gamma$-heavy Fourier coefficient is at most $\norm{\hat{f}}_2 ^2 /\gamma$.
 
The Fourier spectrum of a function $f$ can be approximated using the KM
algorithm in $\ell_\infty$ norm. The KM algorithm, based upon a celebrated result by
\citet{goldreich1989hard}, requires \emph{membership query} (MQ) access to $f$ (\textit{i.e.} it
requires an oracle that for every $x\in \{-1,1\}^n$ returns $f(x)$).
\begin{theorem}[KM algorithm]
\label{theo:KM}
	Let $f:\{-1,1\}^n \rightarrow [-1, 1]$ be a real--valued function and let $\epsilon>0$, $\delta>0$. Then, there exists an algorithm with oracle access to $f$ that, with probability at least $1-\delta$, returns a succinctly represented vector $\tilde{f}$ such that $\norm{\hat{f} - \tilde{f}}_\infty \leq \epsilon$ and $\norm{\tilde{f}}_0 \leq 4/\epsilon^2$. The algorithm runs in $\tilde{O}(n^2\mathrm{log}(1/\delta)/\epsilon^6)$ time and makes $\tilde{O}(n\mathrm{log}(1/\delta)/\epsilon^6)$ queries to $f$.
\end{theorem}

\subsection{$\mu$--biased Fourier analysis}
\noindent
A \textit{product distribution} $\mathcal{D}_\mu$ over $\{-1,1\}^n$ is characterised by a real vector $\mu \in (-1,1)^n$. Such a distribution $\mathcal{D}$ assigns values to each variable independently, so for $x \in \{-1,1\}^n$ we have $\mathcal{D}_\mu(x) = \prod _{i: x_i = 1} (1+\mu_i)/2  \prod _{i: x_i = -1}  (1-\mu_i)/2$ and $\Ex_\mu[x_i] = \mu_i$. Notice that for $\mu = 0$ one recovers the uniform distribution.
We say that the distribution $\mathcal{D}_\mu$ is \textit{$c$-bounded}, or \textit{constant-bounded}, if $\mu \in [-1+c,1-c]^n$, where $c\in(0,1]$. 

\citet{bahadur1961representation} and \citet{furst1991improved} showed that the
Fourier transform can be extended to product distributions, thus defining the
$\mu$\textit{-biased Fourier transform}. The book by \citet{o2014analysis}
gives a brief introduction to $\mu$--biased Fourier analysis and its
applications. For an inner product $\langle f, g \rangle _\mu =
\Ex_{\mathcal{\mu}} [f(x)g(x)]$, the set of functions $\{ \phi_{\mu,a}
\, | \, a\in \{0,1\}^n\}$, where $\phi_{\mu,a} (x) = \prod_{i:a_i = 1} (x_i -
\mu_i)/\sqrt{1-\mu_i ^2}$ forms an orthonormal basis for the vector space of
real--valued functions on $\{-1, 1\}^n$. In this way every function
$f:\{-1,1\}^n\rightarrow \reals$ can be represented as $f(x) = \sum_{a\in
\{0,1\}^n} \hat{f}_\mu (a) \phi_{\mu,a}(x)$, where $\hat{f}_\mu
(a)=\Ex_\mu[f(x)\phi_{\mu,a}(x)]$. For vectors of $\mu$--biased Fourier
coefficients we extend the same notation introduced for standard Fourier
coefficients. Parseval's identity extends to product distributions
$\Ex_\mu [f^2] = \sum_a \hat{f}_\mu (a) ^2 = \norm{\hat{f}_\mu}_2 ^2$.
This implies that for any $f:\{-1,1\}^n \rightarrow [-1,1]$, $\sum_a
\hat{f}_\mu (a) ^2  \leq 1$.

The KM algorithm has been extended to product distributions by \citet{bellare1991spectral}, \citet{jackson1997efficient} and \citet{kalai2009learning}. We give the version presented in~\cite{kalai2009learning}.
\begin{theorem}[EKM algorithm]
\label{theo:EKM}
Let $f:\{-1,1\}^n \rightarrow [-1,1]$ be a real--valued function and let $\epsilon>0$, $\delta>0$, $\mu \in(-1,1)^n$. Then, there exists an algorithm with oracle access to $f$ that, with probability at least $1-\delta$, returns a succinctly represented vector $\tilde{f}_\mu$ such that $\norm{\hat{f}_\mu - \tilde{f}_\mu}_\infty \leq \epsilon$ and $\norm{\tilde{f}_\mu}_0 \leq 4/\epsilon^2$. The algorithm runs in time polynomial in $n$,$1/\epsilon$ and $\mathrm{log}(1/\delta)$.
\end{theorem}

\subsection{Quantum computation and quantum Fourier transform}
\noindent
A generic $n$-qubit state is a complex vector, also known as the \textit{state vector}, acting on a Hilbert space of dimension $2^n$ equipped with an Hermitian scalar product $\braket{\cdot|\cdot}$. We use the Dirac notation to denote quantum states and write $\ket{\psi}$ to denote the quantum state $\psi$. Given a basis $\{ \ket{b_i}\}_{i\in [2^n]}$ the elements of $\ket{\psi}$ correspond to its projections over the basis elements.  Each element of $\ket{\psi}$ corresponds to a different measurable outcome. 
The probability of measurement outcome $i$ is $p(i) = |\psi_i|^2$, where $\psi_i = \braket{\psi|b_i} \in \mathbb{C}$ is the projection of $\ket{\psi}$ onto $\ket{b_i}$. Let $\ket{\psi}$ and $\ket{\phi}$ be two quantum states, their joint description $\ket{\tau}$ is given by the tensor product of the respective state vectors $\ket{\tau}=\ket{\psi} \otimes \ket{\phi}$. 

The evolution of quantum states is governed by \textit{quantum operators}. Quantum operators acting on a $n$-qubit states are $2^n$ dimensional unitary matrices and are denoted with capital letters. Let $x\in \{0,1\}^n$, the QFT over $\mathbb{Z}_2 ^n$ is defined as $H^{\otimes n} \ket{x} = 2^{-n/2} \sum _{a\in \{0,1\}^n} (-1)^{x\cdot a}\ket{a}$, where $H$ is the Hadamard transform $H = \frac{1}{\sqrt{2}} \bigl[ \begin{smallmatrix}1 & 1\\ 1 & -1\end{smallmatrix} \bigl]$.

We often work in the \textit{computational basis} $\{\ket{i}\}$ where, for an $n$-qubit system, each basis element corresponds to an $n$-bit string. A single qubit system can take two values $\ket{0}$ and $\ket{1}$. When working on the Boolean hypercube $\{-1,1\}^n$ we take $0 \equiv -1$ and $1 \equiv 1$. A \textit{quantum register} is a collection of qubits. Given a Boolean--valued function $f:\{-1,1\}^n \rightarrow \{-1,1\}$ a \textit{quantum membership oracle} $O_f$ is a unitary map that applied on $n+1$ qubits acts as follow: $O_f:\ket{i}\ket{0} \rightarrow \ket{i}\ket{f(i)}$. By combining a membership oracle with the Hadamard transform it is possible to make a phase query $\ket{i}\ket{-} \rightarrow (-1)^{f(i)}\ket{i}\ket{-}$, where $\ket{-} = \frac{1}{\sqrt{2}}(\ket{0} - \ket{1})$. This operation is also known as \textit{phase kickback}. For ease of notation, in the following we will not write explicitly the ancilla register $\ket{-}$.

Given a (continuous) probability distribution whose density is efficiently integrable there exists an efficient technique developed by \citet{grover2002creating} to generate a quantum superposition which approximate the distribution.
\begin{lemma}
\label{lemma:terry}
Let $\mathcal{D}$ be a continuous probability distribution and let $\mathcal{D}^*$ be a discretisation of $\mathcal{D}$ over $\{0,1\}^n$. If there exists an efficient classical algorithm to compute $\int _a ^b \mathcal{D}(x) dx$ for every $a,b \in \mathbb{R}$ then, there exists an efficient quantum algorithm that returns the quantum state 
$$
\ket{\psi} = \sum_i \sqrt{\mathcal{D}^*(i)}\ket{i}.
$$
\end{lemma}

\subsection{PAC learning and quantum PAC learning}
\noindent
A \textit{concept class} $C$ is a set of Boolean functions. Every function $f\in C$ is a \textit{concept}. In the PAC model developed by \citet{valiant1984theory} a learner tries to approximate with high probability an unknown concept $f$ from a training set of $m$ random labelled examples $\{(x,f(x))\}$. The examples are given by an \textit{example oracle} $\mathrm{EX}(f,\mathcal{D})$ that returns an example $(x,f(x))$, where $x$ is randomly sampled from a probability distribution $\mathcal{D}$ over $\{-1,1\}^n$. A learning algorithm $A$ for $C$ takes as input an accuracy parameter $ \epsilon \in (0,1)$, a confidence parameter $\delta \in (0,1)$ and the training set and outputs a hypothesis $h$ that is a good approximation of $f$ with probability $1 - \epsilon$. We say that a concept class $C$ is \textit{PAC-learnable} if, for every $\mathcal{D}$, $f$, $h$, $\delta$, when running a learning algorithm $L$ on $m \geq m_{\mathcal{C}}$ examples generated by $\mathcal{D}$, we have that, with probability at least $1-\delta$, $\Pr_{x\sim \mathcal{D}}[h(x)\neq f(x)]\leq\epsilon$.
PAC theory introduces two parameters to classify the efficiency of a learner. The first one, $m_{C}$, is information-theoretic and determines the minimum number of examples required to PAC-learn the class $C$.  We refer to $m_{C}$ as the \textit{sample complexity} of the concept class $\mathcal{C}$. The second parameter, the \textit{time complexity}, is computational and corresponds to the runtime of the best learner for the class $C$. We say that a concept class is \textit{efficiently} PAC-learnable if the running time of $L$ is polynomial in $n$, $1/\epsilon$ and $1/\delta$.%

Two extensions of the PAC model are relevant for our purposes. In the MQ model the learner has access, in addition to the example oracle $\mathrm{EX}(f,\mathcal{D})$, to a membership oracle $\mathrm{MQ}(f)$ that on input $x$ returns $f(x)$.
In the quantum PAC model, the examples are given by a \textit{quantum example oracle} $\mathrm{QEX}(f,\mathcal{D})$ that returns the superposition $\sum_{x} \sqrt{\mathcal{D}(x)} \ket{x, f(x)}$. It has been proven \citep{bshouty1998learning} that membership queries are strictly more powerful than a quantum example oracle (\textit{i.e.} a quantum example oracle can be simulated by a membership oracle but the converse is not true). When $\mathcal{D}$ is the product distribution we use $\mathrm{EX}(f,\mu)$ and $\mathrm{QEX}(f,\mu)$.

A \textit{DNF formula} $f : \{-1,1 \}^n \rightarrow \{-1,1\}$ is a disjunction
of terms where each term is a conjunction of Boolean literals and a literal is
either a variable or its negation (\textit{e.g.}, $ f(x) =(x_2 \wedge x_3
\wedge \neg x_1) \vee (\neg x_4 \wedge x_1)$). The \textit{size of a DNF} $s$
is the number of terms 

\section{Overview of Feldman's algorithm}
\label{sec:feldman}
\noindent
Our proof of the learnability of DNFs under the product distribution builds on an algorithm 
by \citet{feldman2012learning} that greatly simplified the learnability of DNFs. At the core of Feldman's algorithm lies a result by \citet{kalai2009learning} that shows that DNFs can be approximated by heavy low-degree Fourier coefficients alone. More formally, they proved that, for any $s$-term DNF $f$ and for every function $g: \{-1,1 \}^n \rightarrow [-1,1]$, the distance between $f$ and $g$ (measured as $\Ex_\mu [|f(x) - g(x)|]$) is $\Ex_\mu [|f(x) - g(x)|] \leq (2s +1) \cdot \norm{\hat{f}-\hat{g}}_\infty$. This fact gives a direct learnability condition and avoids an involved boosting procedure to turn a weak learner into a strong one (as in the harmonic sieve algorithm by \citet{jackson1997efficient}). Feldman further refined this fact about DNFs.
\begin{theorem}[Theorem 3.8 in~\citep{feldman2012learning}]
\label{theo:threeight}
Let $c\in(0,1]$ be a constant, $\mu$ be a $c$-bounded distribution and $\epsilon>0$. For an integer $s>0$ let $f$ be an $s$-term DNF. For $d=[\log (s/\epsilon) / \log(2/(2-c))]$ and every bounded function $g: \{-1,1\}^n \rightarrow [-1,1]$,
$$
\Ex_\mu [|f(x) - g(x)|] \leq (2\cdot (2-c)^{d/2} \cdot s + 1) \cdot \Norm{ \hat{f}_\mu (B_d) - \hat{g}_\mu(B_d)}_\infty +4\epsilon. 
$$
\end{theorem}
By this theorem the learnability of DNF reduces to constructing a $g$ that approximates the heavy low-degree Fourier spectrum of $f$. This is exactly the approach followed by Feldman that we now proceed to sketch.

The first step of the procedure is to run the EKM algorithm to estimate the heavy Fourier spectrum of $f$. The EKM algorithm returns a succinct representation of $\hat{f}$ and the learner selects only the coefficients that have degree $\leq d$. This is the only step of the algorithm that requires membership queries and is the subroutine that will be replaced by the quantum EKM algorithm that will be derived in Section~\ref{sec:muspectrum}. 

Once the learner has estimated the Fourier spectrum of $f$, it proceeds with the construction of $g$. The procedure is simple and based on an iterative process. Note that by Parseval 
\begin{equation}
\label{eq:exParseval}
\Ex_\mu[(f-g)^2] = \sum_b (\hat{f}_\mu(b)-\hat{g}_\mu(b))^2 = \norm{\hat{f}_\mu-\hat{g}_\mu}^2 _2.
\end{equation}
Suppose there exists and $a$ such that $|\hat{f}_\mu(a) - \hat{g}_\mu(a)| \geq \gamma $. It is possible to construct a $g'$ such that $g'$ is closer than $g$ to $f$ in $l_2$ norm with the following rule:
$$
g' = g + (\hat{f}_\mu(a) - \hat{g}_\mu(a)) \phi_{\mu,a}.
$$
Then by Eq.~\ref{eq:exParseval} we have that
\begin{align*}
\Ex_\mu[(f-g')^2] &= \sum_{b\neq a} (\hat{f}_\mu(b)-\hat{g}_\mu(b))^2 \\
&= \Ex_\mu[(f-g)^2] -  (\hat{f}_\mu(a)-\hat{g}_\mu(a))^2 \\
&\leq \Ex_\mu[(f-g)^2] - \gamma^2.
\end{align*}
The problem with this procedure is that the function $g'$ might have value outside $[-1,1]$ but Feldman showed that the function can be adjusted to the right range and made closer to $f$ in $\ell_2$ distance by cutting-off all the values outside of $[-1,1]$.

Once a precision has been reached such that an application of
Theorem~\ref{theo:threeight} gives $\Ex_\mu [|f(x) - g(x)|] \leq
\epsilon$, the algorithm outputs $\mathrm{sign}(g)$ as hypothesis. From this,
we get the following in regards to learning $f$,
$$
\Pr_\mu [f \neq \mathrm{sign}(g)] \leq \Ex_\mu [|f-g|] \leq \epsilon. 
$$
The running time of all the above operations is polynomial in $n$ and inverse polynomial in the error parameters resulting in the following corollary
\begin{corollary}[Corollary~5.1 in~\citep{feldman2012learning}]
Let $f$ compute an $s$-term DNF. Let $c\in(0,1]$ be a constant and let $\mathcal{D}_\mu$ be a $c$-bounded probability distribution. Let $\mathrm{EX}(f,\mu)$ be an example oracle and $\mathrm{MQ}(f)$ a membership oracle.  Then, there exists an algorithm with $\mathrm{EX}(f,\mu)$ and $\mathrm{MQ}(f)$ access that efficiently PAC learns $f$ over $\mathcal{D}_\mu$.
\end{corollary}
Finally, we note that the requirement of $c$-bounded distributions is imposed in order to control the magnitude of  modulus of the $\mu$-biased Fourier basis $\{|\phi_{\mu,a}|\}$ that, otherwise, would diverge for $\mu$ close to $+ 1$ or $-1$.

\section{Quantum $\mu$--biased Fourier transform}
\label{sec:muQFT}
\noindent
In this section we introduce the $\mu$--biased quantum Fourier transform and show how this can be used to derive a quantum algorithm for sampling from the probability distribution defined by the Fourier coefficients of the $\mu$--biased transform. We recall that the $\mu$--biased Fourier transform is defined as 
\begin{equation}
\label{eq:muQFT}
f(x) = \sum_{a\in \{0,1\}^n} \hat{f}_\mu (a) \phi_{\mu,a}(x),
\end{equation}
where $\phi_{\mu,a} (x) = \prod_{i:a_i = 1} (x_i - \mu_i)/\sqrt{1-\mu_i ^2}$, $\hat{f}_\mu (a)=\Ex_\mu[f(x)\phi_{\mu,a}(x)]$, and $\mathcal{D}_\mu(x) = \prod _{i: x_i = 1} (1+\mu_i)/2  \prod _{i: x_i = -1}  (1-\mu_i)/2$.  Our construction of the $n$-qubit $\mu$--biased QFT exploits a fundamental property of product distributions, namely that the orthonormal basis $\{\phi_{\mu,a}\}$ it defines can be factorised on the individual bits. This fact allows us to give an explicit form of the $n$-qubit transform as a tensor product of $n$ single qubit transforms. We begin by constructing the single qubit transform. Later we will show how to construct efficiently an $n$-qubit transform out of $n$ single qubit ones.  In the following we assume that the function $f:\{-1,1\}^n \rightarrow \{-1,1\}$ is Boolean--valued. Our results can be extended to real--valued functions over the Boolean hypercube using a discretisation procedure. As shown in~\citep{bshouty1998learning} the error induced by the approximation can be controlled. 

Let $b \in \{-1,1\}$ and $v \in \{0,1\}$. The action of the single qubit $\mu$-biased QFT can be explicitly constructed (the normalisation follows from noticing that $\{ \phi_{\mu,v} \}$ forms an orthonormal basis)
\begin{equation}
H_\mu \ket{b} = \sum_{v\in \{0,1\}}  \sqrt{\mathcal{D}_\mu(b)} \phi_{\mu,v}(b) \ket{b}.
\end{equation}
Here we defined $H_\mu$ as the single qubit $\mu$--biased QFT operator whose description in the computational basis is readily given by:
$$
H_\mu = \begin{bmatrix}
       \sqrt{\mathcal{D}_\mu(-1)} \phi_{\mu, 0} (-1)  &   \sqrt{\mathcal{D}_\mu(1)} \phi_{\mu, 0} (1)          \\[0.5em]
        \sqrt{\mathcal{D}_\mu(-1)} \phi_{\mu, 1} (-1) &  \sqrt{\mathcal{D}_\mu(1)} \phi_{\mu, 1} (1)
     \end{bmatrix}.
$$
By taking the functional forms of $\mathcal{D}_\mu(x)$ and $\phi(x)$ we can write
$$
H_\mu = \begin{bmatrix}
        \sqrt{\frac{1-\mu}{2}} &  \sqrt{\frac{1+\mu}{2}}           \\[0.5em]
        -\frac{(1+\mu)\sqrt{1-\mu}}{\sqrt{2-2\mu^2}} & -\frac{(-1+\mu)\sqrt{1+\mu}}{\sqrt{2-2\mu^2}}
      \end{bmatrix}.
$$
It is easy to verify that this matrix is unitary and positive semidefinite. We also note that, as consequence of the Solovay-Kitaev theorem~\citep{kitaev1997quantum}, it is possible to approximate $H_\mu$ to accuracy $\epsilon$ (in operator norm) using $\Theta (\log ^c (1/\epsilon)$ gates from a fixed finite set of universal gates ($c$ is a constant approximately equal to $2$.)

We can construct the extension of the $\mu$--biased QFT to the case of $n$ qubits by taking the tensor product of $n$ single qubit operators. Let $x\in \{-1,1\}^n$ and $a \in \{0,1\}^n$, if we denote as $a_i\in \{0,1\}$ the $i$-th bit of $a$, $\mathcal{D} _{\mu_i} (x) $ as the probability associated to the $i$-th bit, and $\phi_{\mu,a_i}(x)$ its respective basis element, we can write:
$$
H_\mu \otimes \dots \otimes H_\mu \ket{x} = \sum_{a_1}\dots \sum_{a_n} \prod_{i=1} ^n \sqrt{\mathcal{D}_{\mu_{i}}(x)} \phi_{\mu_i,a_i}(x) \ket{a_1}\dots \ket{a_n}. 
$$
By exploiting the product structure of $\mathcal{D}_\mu$ and  $\{\phi_{\mu,a}\}$ that is, $\mathcal{D}_\mu (x)  = \prod_i \mathcal{D} _{\mu_i} (x) $ and $\{\phi_{\mu,a}(x) = \prod_i \phi_{\mu_i,a_i}(x) \}$  we can write the $n$ qubit $\mu$--biased QFT as:
\begin{equation}\label{eq:muQFT}
H_\mu ^{n} \ket{x} = \sum_{a\in \{0,1\}^n}  \sqrt{\mathcal{D}_\mu(x)} \phi_{\mu,a}(x) \ket{a}.
\end{equation}
We remark that it is possible to construct the $n$ qubit transform only because the product distribution and the $\{\phi_{\mu,a}\}$ basis factorise. Without this factorisation we could still write Eq.~\ref{eq:muQFT} but we would not know how to implement this transformation efficiently on a quantum computer (the Solovay-Kitaev theorem guarantees that only single qubit unitary can be efficiently approximated by a universal set of gates).

Finally, we note that the construction of the $\mu$-biased transform assumes knowledge of the vector $\mu$. It is possible to estimate $\mu_i$ for each $i$ using random samples from $\mathcal{D}_\mu$. In Section~\ref{app:error}, we prove that the error introduced by this approximation can be controlled if $\mathcal{D}_\mu $ is $c$-bounded.

As a simple application of the $\mu$--biased QFT, we show how to sample from the
probability distribution defined by the coefficients of the single bit
$\mu$--biased Fourier transform (recall that Parseval's equality holds in the
$\mu$--biased setting).
\begin{lemma}[$\mu$--biased quantum Fourier sampling]
\label{lemma:muQSAMP}
Let $f:\{-1,1\} \rightarrow \{-1,1\}$ be a Boolean--valued function. Then, there exists a quantum algorithm with quantum membership oracle $O_f$ and $H_\mu$ access that returns $v\in \{0,1\}$ with probability $\hat{f}_\mu ^2(v)$. The algorithm requires exactly $1$ $O_f$ query and $3$ gates.
\end{lemma}
\begin{proof}
{Let $f'(x)=(1-f(x))/2$ be the truth table representation of $f(x)$ with $(-1)^{f'(x)}=f(x)$.
Apply Lemma~\ref{lemma:terry} to get $\sum_{b} \sqrt{\mathcal{D}_\mu(b)} \ket{b}$. By querying the quantum membership oracle $O_{f'}$ (given access to $O_{f}$ this is equivalent to a relabelling of the qubits) one can make a phase query and obtain $\sum_{b} \sqrt{\mathcal{D}_\mu(b)} f(b) \ket{b}$ (note that $|f(b)|=1$ and therefore the state is still normalised). Finally, applying the $\mu$--biased QFT results in
\begin{align*}
  \sum_{b} \sqrt{\mathcal{D}_\mu(b)}f(b) \left( \sum_v \sqrt{\mathcal{D}_\mu(b)} \phi_{\mu,v} (b) \ket{v} \right) &= \sum_{b,v}\mathcal{D}_\mu(b) f(b) \phi_{\mu,v} (b) \ket{b}  \\
  &= \sum_v \hat{f}_\mu (v) \ket{v}.
\end{align*} 
Measuring the state, one obtains $v$ with probability $\hat{f}_\mu ^2 (v)$}
\end{proof}

In order to use this result in the context of quantum PAC learning we need to replace the membership oracle $O_f$ with a quantum example oracle. The following lemma, that extends Lemma~$1$ in~\citep{bshouty1998learning} to the $\mu$-biased case, serves this purpose. Differently from Lemma~\ref{lemma:muQSAMP} we present directly the $n$-dimensional case.
\begin{lemma}
\label{lemma:FourierEst}
Let $f:\{-1,1\}^n \rightarrow \{-1,1\}$ be a Boolean-valued function. Then, there exists a quantum algorithm with quantum example oracle $\mathrm{QEX}(f,\mu)$ access that returns $a\in \{0,1\}^n$ with probability $\hat{f}_\mu ^2(a)/2$. The algorithm requires exactly $1$ $\mathrm{QEX}$ query and $O(n)$ gates.
\end{lemma}
\begin{proof}
{Let $f'(x)=(1-f(x))/2$ be the truth table representation of $f(x)$ with $(-1)^{f'(x)}=f(x)$. Given access to $\mathrm{QEX}(f,\mu)$ it is always possible to construct an oracle for $\mathrm{QEX}(f',\mu)$ (this is equivalent to a relabelling of the qubits). Apply $\mathrm{QEX}(f',\mu)$ on a $\ket{0,\dots,0,0}$ to get $\sum_{x} \sqrt{\mathcal{D}_\mu(x)} \ket{x, f'(x)}$. Then apply $H_\mu ^n$ on the first register:
\vskip -0.7em
  $$
  \sum_{x\in \{-1,1\}^n} \sum_{a\in \{0,1\}^n} \sqrt{\mathcal{D}_\mu(x)} \sqrt{\mathcal{D}_\mu(x)} \phi_{\mu,a} (x) \ket{a, f'(x)}.
  $$
An application of the standard QFT on the second register gives:
\begin{align*}
\sum_{x,a,z} \frac{1}{\sqrt{2}} (-1)^{f'(x)z} \mathcal{D}_\mu(x) \phi_{\mu,a} (x) \ket{a, z} &=  \frac{1}{\sqrt{2}}\left( \sum_a \hat{f}_\mu(a) \ket{a, 1} + \sum_a \Ex_{\mu} [\phi_{\mu,a} (x) ] \ket{a,0} \right) \\
&= \frac{1}{\sqrt{2}}\left( \sum_a \hat{f}_\mu(a) \ket{a, 1} + \sum_a \Ex_{\mu} [\phi_{\mu,a} (x) \phi_{\mu,0} (x)] \ket{a,0} \right) \\
&= \frac{1}{\sqrt{2}}\left( \sum_a \hat{f}_\mu(a) \ket{a, 1} +  \ket{0\dots0,0} \right), 
\end{align*}
where we used the orthonormality of the $\{\phi_{\mu,a}\}$ basis and $\phi_{\mu,0} (x) = 1$. Measuring the first register we obtain $\ket{a,1}$ with probability $\hat{f}_\mu ^2 (a)/2$}
\end{proof}

\section{Quantum computation of $\mu$--biased Fourier spectrum}
\label{sec:muspectrum}
\noindent
In this section we give a quantum algorithm to approximate the $\mu$-biased Fourier spectrum of a function. This can be interpreted as a quantum version of the EKM algorithm. As a simple application of the quantum EKM algorithm we obtain the learnability of DNFs under product distributions in the quantum PAC model.

Our proof uses the \textit{Dvoretzky-Kiefer-Wolfowitz} (DKW) theorem, a concentration inequality that bounds the number of samples required to estimate a cumulative distribution in $\ell_\infty$ norm. The DKW Theorem was first proposed by \citet{dvoretzky1956asymptotic} with an almost tight bound. \citet{birnbaum1958distribution} conjectured that the inequality was tight. This conjecture was proved by \citet{massart1990tight}. The DKW theorem is usually given for continuous probability distribution but its validity extends also to discrete distributions (a detailed discussion can be found in~\citep{kosorok2007introduction}).

Let $X_1,\dots,X_m$ be a sequence of i.i.d. random variables drawn from a distribution $f$ on $\mathbb{R}$ with \textit{Cumulative Distribution Function} (CDF) $F(x)= \sum_{X_i \leq x} f(X_i) $, and let $x_1, \dots, x_n$ be their realisations. Given a set $A$ the indicator function $\mathbf{1}_{A} : A \rightarrow \{0,1\}$ takes values $f(x) = 0$ if $x \notin A$ and $f(x) = 1$ if $x \in A$. We denote the empirical probability distribution associated to $f(x)$ as $ f_m(x) = \sum_{i=1} ^m \mathbf{1}_{\{X_i = x \}}/m $ and its empirical cumulative distribution as $ F_m (x) = \frac{1}{m}\sum_{i=1}^{m} \mathbf{1}_{\{X_i \leq x \}} $. We present a version of the DKW theorem adapted from Theorem $11.6$ in \cite{kosorok2007introduction}.

\begin{theorem}[Dvoretzky-Kiefer-Wolfowitz]
\label{theo:DKW}
Let $X_1,\dots,X_m$ be a sequence of i.i.d. random variables with cumulative distribution $F$ and empirical cumulative distribution defined by $F_m (x) = \frac{1}{m}\sum_{i=1}^{m} \mathbf{1}_{\{X_i \leq x \}}$. Then for any $\epsilon \geq 0$,
$$
\mathrm{Pr}\left(\sup_{x\in \mathbb{R}} |F (x) - F_m(x)| > \epsilon \right) \leq 2 e^{-2m\epsilon^2},
$$
for all $\epsilon>0$.
\end{theorem}

We note that DKW theorem holds also to the case where $F$ is discontinuous (the discontinuities can be infinite but must be countable). A proof of this result is presented in \citep{kosorok2007introduction}. Therefore, the DKW theorem can be applied to the case of discrete random variables and make notation consistent. In this case, it is possible to replace the $\sup$ with the $\max$ and we write $\norm{F (x) - F_m(x)}_\infty$ instead of $\max_{x\in \{-1,1\}^n} |F (x) - F_m(x)|$. 

By using the DKW theorem we can prove a useful lemma that bounds the number of samples needed to estimate a probability distribution in $\ell_\infty$ norm.
\begin{lemma}
\label{lemma:piest}
Let $f$ be a probability distribution over $\{-1,1\}^n$ and let $\tau>0$, $\delta>0$. Then, there exits an algorithm that with probability $1-\delta$ and for $m = O(\mathrm{log}(1/\delta)/\tau^2)$ outputs $f_m$ such that $\norm{f-f_m}_\infty \leq \tau$.
\end{lemma}
\begin{proof}
{Let $\{e_1, \dots, e_{2^n}\}$ be an ordering of elements of the Boolean hypercube $\{-1,1\}^n$. We have that
\begin{align*}
\norm{f-f_m }_\infty &= \max_{\{e_1, \dots, e_{2^n}\}} |f(e_i)-f_m (e_i)| \\
&= \max_{\{e_1, \dots, e_{2^n}\}} |F(e_{i+1})-F_m(e_{i+1})-(F(e_{i})-F_m(e_{i}))|.
\end{align*}
An application of the triangle inequality gives
\begin{align*}
\max_{\{e_1, \dots, e_{2^n}\}} | & F(e_{i+1})-  F_m(e_{i+1}) -(F(e_{i})-F_m(e_{i}))|  \\ 
& \leq \max_{\{e_1, \dots, e_{2^n}\}} |F(e_{i+1})-F_m(e_{i+1})|  +\max_{\{e_1, \dots, e_{2^n}\}} |F(e_{i})-F_m(e_{i})| \\
& \leq 2\, \norm{F-F_m }_\infty.
\end{align*}
By Theorem~\ref{theo:DKW} we have that, with probability $1-\delta$, 
\[
\mathrm{Pr}(\norm{F  - F_m}_\infty \geq \gamma) \leq 2 e^{-2m\gamma^2}.
\]
Let $\gamma = \tau/2$, then
$$
\mathrm{Pr}(\norm{f - f_m }_\infty \leq \tau) \leq 1 - 2 e^{-m\tau^2 /2},
$$
from which it is easy to see that $m = O(\mathrm{log}(1/\delta)/\tau^2)$}
\end{proof}
The combined application of Lemma~\ref{lemma:FourierEst} and Lemma~\ref{lemma:piest} allows us to prove the following result:
\begin{theorem}[Quantum EKM algorithm]
\label{theo:qKM}
Let $f:\{-1,1\}^n \rightarrow \{-1,1\}$ be a Boolean--valued function and let $\epsilon>0$, $\delta>0$, $\mu \in(-1,1)^n$. Then, there exists a quantum algorithm with $\mathrm{QEX}(f,\mu)$ access that, with probability at least $1-\delta$, returns a succinctly represented vector $\tilde{f}_\mu$, such that $\norm{\hat{f}_\mu - \tilde{f}_\mu}_\infty \leq \epsilon$ and $\norm{\tilde{f}_\mu}_0 \leq 4/\epsilon^2$. The algorithm requires $O(\mathrm{log}^2(1/\delta)/\epsilon^8)$ $\mathrm{QEX}(f,\mu)$ queries and $O(n\, \mathrm{log}^2(1/\delta)/\epsilon^8)$ gates.
\end{theorem}
\begin{proof}
{We begin by estimating the $a$'s corresponding to the $\epsilon/2$-heavy Fourier coefficients of $f$. 
Let $\{ p(a) = |\hat{f}_\mu (a)|^2 \}$ be the probability distribution defined by the $\mu$--biased Fourier coefficients of $f$. Lemma~\ref{lemma:FourierEst} gives a procedure that, with $1$ $\QEX(f,\mu)$ query and $O(n)$ gates, measures $\ket{a,1}$ with probability $q(a,1) = |\hat{f}_\mu (a)|^2 /2 $ and $\ket{0\dots 0,0}$ with probability $q(0,0) = 1/2$.
Applying Lemma~\ref{lemma:piest} on the distribution $q$ with $\tau = \epsilon^2 / 8 $ we obtain that $O(\mathrm{log}(1/ \delta) \epsilon^4)$ samples are required to have an estimate $\norm{q - \tilde{q}}_\infty \leq \epsilon^2 / 8$ with high probability. This implies that $\norm{\hat{f}_\mu ^2 - \tilde{f}_\mu ^2}_\infty \leq \epsilon^2 / 4$. By selecting the characteristic vectors that correspond to coefficients such that $|\tilde{f}_\mu (a) |^2 > \epsilon^2 /2$ (and discarding the element $\ket{a,0}$)  we can output a list of $a$'s such that, with probability $\geq 1-\delta$, all the corresponding Fourier coefficients have $|\hat{f}_\mu (a)| > \epsilon$ and there are no coefficients such that $|\hat{f}_\mu (a)| \leq \epsilon / 2$. By Parseval's equality this implies that the list may contain at most $4/ \epsilon^2$ elements.

The final step requires the estimation of the Fourier coefficients. For a given $a$, the Fourier coefficient $\hat{f}_\mu (a) = \Ex_\mu[f(x)\chi_a (x)]$ can be obtained by sampling using the $\QEX(f,\mu)$ oracle to simulate $\EX(f,\mu)$ (to get an example $(x,f(x))$ it would suffice to measure a state prepared with $\QEX(f,\mu)$) in time $O(\mathrm{log}(1/\delta)/\epsilon^2)$ (the number of examples required for the estimate is a standard application of the Hoeffding bound). 

The total number of examples required to estimate all the $\epsilon/2$-heavy Fourier coefficients of $f$ is  $O(t \mathrm{log}^2(1/\delta)/\epsilon^8)$ by the union bound, where $t$ is the number of $\epsilon/2$-heavy Fourier coefficients. Because by Parseval's $t \leq 4/\epsilon^2 $ we have that the final algorithm requires $O(\mathrm{log}^2(1/\delta)/\epsilon^8)$ $\QEX(f,\mu)$ queries and $O(n\, \mathrm{log}^2(1/\delta)/\epsilon^8)$ gates}
\end{proof}

Theorem~\ref{theo:qKM} can be straightforwardly used in the method developed by \citet[Corollary~5.1]{feldman2012learning} to obtain the learnability of DNF under product distributions.  
\begin{corollary}
Let $f$ compute an $s$-term DNF. Let $c\in(0,1]$ be a constant, let $\mathcal{D}_\mu $ be a $c$-bounded probability distribution and let $\mathrm{QEX}(f,\mu)$ be a quantum example oracle.  Then, there exists a quantum algorithm with $\mathrm{QEX}(f,\mu)$ access that efficiently PAC learns $f$ over $\mathcal{D}_\mu$.
\end{corollary}
We recall that the collection of the heavy Fourier coefficients of the DNF $f$ is the only step of Feldman's algorithm that requires MQ. The remaining of the algorithm makes use of the coefficients to construct a function $g$ that approximates $f$.

\section{Error analysis}
\label{app:error}
\noindent
In the main text we assumed that the vector $\mu$ parametrising the product distribution was given to the learner. Here we prove that, if $\mathcal{D}_\mu$ is $c$-bounded, it is possible to estimate $\mu$ introducing an error that can be made small at a cost that scales polynomially in $n$. We recall that $\mu \in [-1+c,1-c]^n$, $c\in (0,1]$, and $\mu_i = \Ex_\mu [x_i]$. 
By the Hoeffding bound we have that, with probability $1-\delta$, it is possible to approximate $\mu_i$ to $\epsilon$ accuracy $|\mu_i - \tilde{\mu_i}| \leq \epsilon$ using $O(\mathrm{log}(1/\delta) / \epsilon^2)$ samples.

We want to estimate the error introduced by approximating $H^n _\mu$ with $H_{\tilde{\mu}} ^n $ (note that the $\mu$-biased QFT is now parametrised by $\tilde{\mu}$) in terms of the operator norm. Let $A$ be an operator, the operators norm $\Norm{A}$ is defined as:
$$
\Norm{A} = \sup_{\ket{\psi} \neq 0} \frac{\Norm{A\ket{\psi}}}{\Norm{\ket{\psi}}}.
$$
The error analysis then requires to bound the quantity:
$$
\Norm{H^n _\mu - H_{\tilde{\mu}} ^n } \leq \gamma.
$$
In order to prove the bound we introduce a useful lemma:
\begin{lemma}
\label{lemma:prodUnit}
Let $A = A_n \cdots A_1$ be a product of unitary operators $A_j$. Assume that for every $A_j$ there exits an approximation $\tilde{A}_j$ such that $\norm{A_j-\tilde{A}_j} \leq \epsilon_j$. The follow inequality holds
$$
\Norm{A_n \cdots A_1 -  \tilde{A}_n \cdots \tilde{A}_1 } \leq \sum_j \epsilon_j.
$$
\end{lemma}
\begin{proof}
{We prove by induction. The base step follows from the assumptions. For the inductive step let $X_{k} = A_k \cdots A_1 $ and $\tilde{X}_{k} = \tilde{A}_k \cdots \tilde{A}_1$. Because the inductive hypothesis holds we have
$$
\norm{X_{k} - \tilde{X}_{k}} \leq \sum_{j=1} ^k \epsilon_j.
$$
By making use of the triangular inequality, the induction hypothesis, and noting that the product of unitaries is unitary we have
\begin{align*}
\Norm{A_{k+1} X_k - \tilde{A}_{k+1} \tilde{X}_k} &= \Norm{A_{k+1}\left( X_k - \tilde{X}_k \right) + \left( A_{k+1} - \tilde{A}_{k+1} \right) \tilde{X}_k } \\
&\leq \Norm{A_{k+1}\left( X_k - \tilde{X}_k \right)} + \Norm{\left( A_{k+1} - \tilde{A}_{k+1} \right) \tilde{X}_k }\\
&= \Norm{A_{k+1}} \Norm{ X_k - \tilde{X}_k } + \Norm{A_{k+1} - \tilde{A}_{k+1}} \Norm{\tilde{X}_k } \\
&\leq \sum_{j=1} ^k \epsilon_j + \epsilon_{k+1} \\
&= \sum_{j=1} ^{k+1} \epsilon_j
\end{align*}
}
\end{proof}
Let $H_j = I \otimes \cdots I \otimes H_{\mu_j} \otimes I \otimes \cdots I$ and $\tilde{H}_j = I \otimes \cdots I \otimes H_{\tilde{\mu}_j} \otimes I \otimes \cdots I$. By Lemma~\ref{lemma:prodUnit} we have that
\begin{equation}
\label{eq:simpbound}
\Norm{H^n _\mu - H_{\tilde{\mu}} ^n } \leq \sum_{j=1} ^n \Norm{H_j - \tilde{H}_j}.
\end{equation}
The bound on $\Norm{H_j - \tilde{H}_j}$ can be simplified using the following property of the operator norm
$\norm{A\otimes B}= \norm{A}\norm{B}$,
\begin{align*}
\Norm{H_j - \tilde{H}_j} &= \Norm{I \otimes \cdots I \otimes (H_{\mu_j} - H_{\tilde{\mu}_j}) \otimes I \otimes \cdots I }\\
&= \norm{I} \cdots \Norm{I}\Norm{(H_{\mu_j} - H_{\tilde{\mu}_j})} \Norm{I}  \cdots \Norm{I }\\
&= \Norm{H_{\mu_j} - H_{\tilde{\mu}_j}}.
\end{align*}
The problem of bounding Eq.~\ref{eq:simpbound} is then equivalent to bounding $\norm{H_{\mu_i}- H_{\tilde{\mu}_i}}$. Let $\ket{\psi}=\sum_{x \in \{-1,1\} } \alpha_x \ket{x}$, we have that
\begin{align*}
\Norm{(H_{\mu_i}- H_{\tilde{\mu}_i})\ket{\psi}} = \Norm{\sum_{x \in \{-1,1\}} \sum_{a \in \{0,1\}} \left(\sqrt{\mathcal{D}_{\mu_i} (x)}\phi_{\mu_i,a}(x) - \sqrt{\mathcal{D}_{\tilde{\mu}_i} (x)}\phi_{\tilde{\mu}_i,a}(x) \right) \alpha_x \ket{a}}
\end{align*}
where $\phi_{\mu,a} (x) = \prod_{i:a_i = 1} (x_i - \mu_i)/\sqrt{1-\mu_i ^2}$ and $\mathcal{D}_\mu(x) = \prod _{i: x_i = 1} (1+\mu_i)/2  \prod _{i: x_i = -1}  (1-\mu_i)/2$.
We have to estimate the following quantity for a generic $a$, $x$
\begin{align*}
S &= \left| \sqrt{\mathcal{D}_{\mu_i} (x)}\phi_{\mu_i,a}(x) - \sqrt{\mathcal{D}_{\tilde{\mu}_i} (x)}\phi_{\tilde{\mu}_i,a}(x) \right | \\
&= \left| \frac{(x_i - \mu_i)\sqrt{1-\tilde{\mu}_i ^2}\sqrt{\mathcal{D}_{\mu_i}(x)} - (x_i - \tilde{\mu}_i)\sqrt{1-\mu_i ^2}\sqrt{\mathcal{D}_{\tilde{\mu}_i}(x)} }{\sqrt{1-\mu_i ^2}\sqrt{1-\tilde{\mu}_i ^2} }\right |.
\end{align*}
Recall that for every $i$ it holds $1-\mu_i ^2 \geq c^2$, $1-\tilde{mu}_i ^2 \geq c^2$, $|\mu_i - \tilde{\mu}_i| \leq \epsilon$, $x_i \in \{-1,1\}$. By the triangle inequality we have that
\[
\begin{split}
S &\leq \frac{1}{c^2} \left| (x_i - \mu_i)\sqrt{1-\tilde{\mu}_i ^2}\sqrt{\mathcal{D}_{\mu_i}(x)} - (x_i - \tilde{\mu}_i)\sqrt{1-\mu_i ^2}\sqrt{\mathcal{D}_{\tilde{\mu}_i}(x)} \right | \\
&= \frac{1}{c^2} \Big| (x_i - \mu_i)\left( \sqrt{1-\tilde{\mu}_i ^2}\sqrt{\mathcal{D}_{\mu_i}(x)}  - \sqrt{1-\mu_i ^2}\sqrt{\mathcal{D}_{\tilde{\mu}_i}(x)}\right) \\ 
&\phantom{{}={}} \qquad + (\tilde{\mu}_i - \mu_i) \sqrt{1-\mu_i ^2}\sqrt{\mathcal{D}_{\tilde{\mu}_i}(x)} \Big | \\
&\leq \frac{1}{c^2} \Big( \left| (x_i - \mu_i)\left( \sqrt{1-\tilde{\mu}_i ^2}\sqrt{\mathcal{D}_{\mu_i}(x)} - \sqrt{1-\mu_i ^2}\sqrt{\mathcal{D}_{\tilde{\mu}_i}(x)}\right)\right | \\
&\phantom{{}={}} \qquad + \left|(\tilde{\mu}_i - \mu_i) \sqrt{1-\mu_i ^2}\sqrt{\mathcal{D}_{\tilde{\mu}_i}(x)} \right | \Big) \\
&\leq \frac{1}{c^2} \left( (2-c)\left|\left( \sqrt{1-\tilde{\mu}_i ^2}\sqrt{\mathcal{D}_{\mu_i}(x)} - \sqrt{1-\mu_i ^2}\sqrt{\mathcal{D}_{\tilde{\mu}_i}(x)}\right)\right | + \epsilon \right) \\
&\leq \frac{1}{c^2} \left( (2-c)\left( \left| \sqrt{\mathcal{D}_{\mu_i}(x)} - \sqrt{\mathcal{D}_{\tilde{\mu}_i}(x)} \right| +  \left| \sqrt{1-\mu_i ^2} - \sqrt{1-\tilde{\mu}_i ^2} \right|\right ) + \epsilon \right).
\end{split}
\]
If we note that 
$$
\left| \sqrt{\mathcal{D}_{\mu_i}(x)} - \sqrt{\mathcal{D}_{\tilde{\mu}_i}(x)} \right| = \left|\frac{\tilde{\mu_i} - \mu_i}{2(\sqrt{\mathcal{D}_{\mu_i}(x)} + \sqrt{\mathcal{D}_{\tilde{\mu}_i}(x)} )} \right| \leq \frac{\epsilon}{\sqrt{8c}}
$$
and 
$$
\left|\sqrt{1-\mu_i ^2} - \sqrt{1-\tilde{\mu}_i ^2} \right| = \left|\frac{\tilde{\mu}_i ^2 - \mu_i ^2}{\sqrt{1-\mu_i ^2} + \sqrt{1-\tilde{\mu}_i ^2} }\right| \leq \frac{\epsilon}{2c}
$$ we have
\begin{equation}
\label{eq:finbound}
S \leq t \epsilon,
\end{equation}
where $t = ((2 - c) (\frac{1}{\sqrt{8 c}} + \frac{1}{ 2 c}) + 1 )/c^2$. By making use of Eq.~\ref{eq:finbound} and noting that $\sum_x |\alpha_x| \leq 1$ we have
\begin{align*}
\norm{H_{\mu_i} - H_{\tilde{\mu}_i} } &\leq \sum_{x,a} \Norm{ \left(\sqrt{\mathcal{D}_{\mu_i} (x)}\phi_{\mu_i,a}(x) - \sqrt{\mathcal{D}_{\tilde{\mu}_i} (x)}\phi_{\tilde{\mu}_i,a}(x) \right) \alpha_x \ket{a}} \\
&\leq t \epsilon \left(2\sum_x |\alpha_x| \right) \\
&\leq 2t \epsilon.
\end{align*}
From which it follows that 
\begin{equation}
\label{eq:lastbound}
\norm{H_\mu ^n - H_{\tilde{\mu}} ^n } \leq 2nt \epsilon.
\end{equation}
Eq.~\ref{eq:lastbound} guarantees that if one can approximate every $\mu_i$ with linear precision, \textit{i.e.} with $\epsilon = O(1 / n)$, it is possible to control the approximation error. Recall that by using the Hoeffding bound we can approximate $\mu_i$ to linear precision using $m  = O(n^2 \log (1/ \delta))$ examples.

\subsection*{Acknowledgements}

We thank Carlo Ciliberto for helpful discussions on the DKW inequality and Matthias Caro for comments on an earlier draft.
This research was supported in part by the National Science
Foundation under Grant No. NSF PHY-1748958 and by the Heising-Simons
Foundation.
AR is supported by an EPSRC DTP Scholarship, by QinetiQ Ltd., and by the Simons Foundation through \textit{It from Qubit: Simons Collaboration on Quantum Fields, Gravity, and Information}.
SS is supported by the Royal Society, EPSRC, the National Natural Science Foundation of China, and the grant ARO-MURI W911NF-17-1-0304 (US DOD, UK MOD and UK EPSRC under the Multidisciplinary University Research Initiative).

\printbibliography

\end{document}